\documentclass[submission,copyright,creativecommons]{eptcs}
\providecommand{\event}{M4M--9} 
\usepackage{breakurl}             
\usepackage{underscore}           

\usepackage{amsmath,amssymb,amsthm}
\usepackage{colonequals}

\title{Temporal Justification Logic}
\author{Samuel Bucheli
\institute{Zühlke Engineering AG,} 
\institute{ Bogenschützenstrasse 9A,\\ 3008 Bern, Switzerland}
\email{samuel.bucheli.cs@gmail.com}
\and
Meghdad Ghari\thanks{This research was in part supported by a grant from IPM (No. 95030416).}
\institute{School of Mathematics,\\ Institute for Research\\ in Fundamental Sciences (IPM),}
\institute{P.O. Box: 19395-5746, Tehran, Iran}
\email{ghari@ipm.ir}
\and
Thomas Studer\thanks{This work was partially supported by the SNSF project 200021\_165549 \emph{Justifications and non-classical reasoning}.}
\institute{Institut f\"ur Informatik,\\ Universit\"at Bern,} 
\institute{ Neubr\"uckstrasse 10,\\ 3012 Bern, Switzerland}
\email{tstuder@inf.unibe.ch}
}
\def\titlerunning{Temporal Justification Logic}
\def\authorrunning{S. Bucheli, M. Ghari \& T. Studer}

\newcommand{\Prop}{\textsf{Prop}}
\newcommand{\Formulae}{\textsf{Fml}}

\newcommand{\lfalse}{\bot}
\newcommand{\ltrue}{\top}
\newcommand{\lneg}{\neg}

\newcommand{\propax}{\ensuremath{(\textsf{Taut})}}
\newcommand{\lrule}[2]{\displaystyle{\frac{#1}{#2}}}
\newcommand{\mprule}{\ensuremath{(\textsf{MP})}}
\newcommand{\limplies}{\rightarrow}
\newcommand{\liff}{\leftrightarrow}

\newcommand{\lnext}{\bigcirc}
\newcommand{\lalways}{\Box}
\newcommand{\leventually}{\Diamond}
\newcommand{\luntil}{{\,\mathcal{U}\,}}

\newcommand{\kax}{\ensuremath{\textsf{-k}}}
\newcommand{\nextkax}{\ensuremath{(\lnext\kax)}}
\newcommand{\alwayskax}{\ensuremath{(\lalways\kax)}}
\newcommand{\funax}{\ensuremath{(\textsf{fun})}}

\newcommand{\indax}{\ensuremath{(\textsf{ind})}}
\newcommand{\uoneax}{\ensuremath{(\luntil\textsf{1})}}
\newcommand{\utwoax}{\ensuremath{(\luntil\textsf{2})}}
\newcommand{\necrule}{\ensuremath{\textsf{-nec}}}
\newcommand{\nextnecrule}{\ensuremath{(\lnext\necrule)}}
\newcommand{\alwaysnecrule}{\ensuremath{(\lalways\necrule)}}
\newcommand{\uindrule}{\ensuremath{(\luntil\textsf{-R})}}

\newcommand{\LPLTL}{\textsf{LPLTL}}
\newcommand{\LL}{\textsf{LPLTL}^\star}

\newcommand{\lknows}{\mathsf{K}}

\newcommand{\SFive}{\textsf{S5}}

\newcommand{\CTerms}{\textsf{Const}}
\newcommand{\VTerms}{\textsf{Var}}
\newcommand{\Terms}{\textsf{Tm}}

\newcommand{\jbox}[1]{\left[#1\right]\!}
\newcommand{\tapp}{\cdot}
\newcommand{\tsum}{+}
\newcommand{\tinspect}{!}

\newcommand{\tnext}{\Rrightarrow}
\newcommand{\tprev}{\Lleftarrow}
\newcommand{\talwaysaccess}{\Downarrow}
\newcommand{\tgeneralize}{\Uparrow}

\newcommand{\appax}{\ensuremath{(\textsf{application})}}
\newcommand{\sumax}{\ensuremath{(\textsf{sum})}}
\newcommand{\posintax}{\ensuremath{(\textsf{positive introspection})}}

\newcommand{\refax}{\ensuremath{(\textsf{reflexivity})}}
\newcommand{\constnecrule}{\ensuremath{(\textsf{ax}\necrule)}}

\newcommand{\CS}{\textsf{CS}}

\newcommand{\numberofagents}{h}
\newcommand{\agent}{i}

\newcommand{\alwaysaccessprinciple}{\ensuremath{(\lalways\textsf{-access})}}
\newcommand{\generalizeprinciple}{\ensuremath{(\textsf{generalize})}}
\newcommand{\nextrightshiftprinciple}{\ensuremath{(\lnext\textsf{-access})}}
\newcommand{\nextleftshiftprinciple}{\ensuremath{(\lnext\textsf{-left})}}

\newcommand{\runs}{\mathcal{R}}
\newcommand{\system}{\mathcal{I}}

\newcommand{\evidence}{\mathcal{E}}
\newcommand{\valuation}{\nu}
\newcommand{\entails}{\vDash}

\newcommand{\N}{\mathbb{N}}
\newcommand{\powerset}{\mathcal{P}}
\newcounter{enumsave}
\newcommand{\boxrefax}{\ensuremath{(\textsf{boxed reflexivity})}}
\newcommand{\mixprinciple}{\ensuremath{(\textsf{mix})}}

\newtheorem{theorem}{Theorem}
\newtheorem{lemma}[theorem]{Lemma}
\newtheorem{corollary}[theorem]{Corollary}

\newtheorem{definition}[theorem]{Definition}

\newtheorem*{remark}{Remark}
\newtheorem{question}{Question}

\renewcommand{\phi}{\varphi}

\newcommand{\Sub}{\mathsf{Sub}}
\newcommand{\MCS}{\mathsf{MCS}}

\begin{document}
\maketitle

\begin{abstract}
Justification logics are modal-like logics with the additional capability of recording the reason, or justification, for modalities in syntactic structures, called justification terms. 
Justification logics can be seen as explicit counterparts to modal logics.
The behavior and interaction of agents in distributed system is often modeled using logics of knowledge and time.
In this paper, we sketch some preliminary ideas on how the modal knowledge part of such logics of knowledge and time could be replaced with an appropriate justification logic.
\end{abstract}

 \section{Introduction}
 \label{sect:Introduction}

Justification logics
are epistemic logics that feature explicit reasons for an agent's knowledge and belief.
Originally, Artemov~\cite{Art01BSL} developed the first justification logic, the Logic of Proofs~$\mathsf{LP}$, to provide a  classical provability semantics for intuitionistic logic.
Later, Fitting~\cite{Fit05APAL} introduced epistemic models for justification logic.
This general reading of justification led to a big variety of  epistemic justification logics for many different applications~\cite{Art06TCS,Art08RSL,
ArtKuz14APAL,BucKuzStu11JANCL,Ghari:2014,Gha16JIGPL,KMOS15,KuzStu13LFCS}.
Instead of an implicit statement~$\lknows \phi$, which stands for  \emph{the agent knows~$\phi$}, justification logics include explicit statements of the form~$\jbox{t} \phi$, which mean \emph{$t$~justifies the agent's knowledge of~$\phi$}.

 A common approach to model distributed systems of interacting agents is using logics of knowledge and time, with the interplay between these two modalities leading to interesting properties and questions~\cite{FHMV95,HvdMV04,HZ92,vdMW03,vDvdHR13}.
 While knowledge in such systems has typically been modeled using the modal logic $\SFive$, it is a natural question to ask what happens when we model knowledge in such logics using a justification logic.

This paper offers a first study on combing temporal logic and justification logic. 
We introduce a system $\LPLTL_\CS$  that combines linear time temporal logic $\mathsf{LTL}$ with the justification logic $\mathsf{LP}$. In Sections~\ref{sect:Syntax} and~\ref{sect:Axioms} we present the language and the axioms of $\LPLTL_\CS$, respectively. In Section~\ref{sect:Semantics} we introduce interpreted systems with Fitting models as semantics for temporal justification logic. In Section~\ref{sect:Completeness} we establish soundness and completeness of $\LPLTL_\CS$. In Section~\ref{sect:internalization} we present an extension $\LL_\CS$ of $\LPLTL_\CS$ that  enjoys the internalization property.
In Section~\ref{sect:addidtional} we introduce some additional principles concerning interactions of knowledge, justifications, and time.
In Section~\ref{s:conclusion} we conclude the paper and discuss some open problems.

\noindent
{\bf Acknowledgments.} We would like to thank the anonymous referees for many helpful comments, which helped to improve the paper.

 \section{Language}
 \label{sect:Syntax}

  In the following, let $\numberofagents$ be a fixed number of agents, $\CTerms$ a countable set of justification constants, $\VTerms$ a countable set of justification variables, and $\Prop$ a countable set of atomic propositions.

 The set of justification terms $\Terms$ is defined inductively by
 \[
  t \coloncolonequals c \mid x \mid \; \tinspect t \mid \;  t \tsum t \mid t \tapp t \, ,
 \]
 where $c \in \CTerms$ and $x \in \VTerms$.

 The set of formulas $\Formulae$ is inductively defined by 
 \[
  \phi \coloncolonequals P \mid \lfalse \mid \phi \limplies \phi \mid \lnext \phi \mid  \phi \luntil \phi \mid \jbox{t}_\agent\phi \, , 
 \]
 where $1\leq\agent\leq\numberofagents$, $t \in \Terms$,  and $P \in \Prop$.

We use the following usual abbreviations:
 \begin{align*}
  \lneg \phi &\colonequals \phi \limplies \lfalse &
  \ltrue &\colonequals \lneg \lfalse &\\
  \phi \lor \psi &\colonequals \lneg \phi \limplies \psi &
  \phi \land \psi &\colonequals \lneg (\lneg \phi \lor \lneg \psi) \\
  \phi \liff \psi &\colonequals (\phi \limplies \psi) \land (\psi \limplies \phi) &
  \Diamond \phi &\colonequals \top \luntil \phi \\
  \Box \phi &\colonequals \lneg \Diamond \lneg \phi.
 \end{align*}

Associativity and precedence of connectives, as well as the corresponding omission of brackets, are handled in the usual manner.

Subformulas are defined as usual. The set of subformulas $\Sub(\chi)$ of a formula $\chi$ is inductively given by:
 \begin{align*}
 \Sub(P) &\colonequals \{P\}  & \Sub(\bot) &\colonequals \{\bot\}\\
 \Sub(\phi \to \psi) &\colonequals \{\phi \to \psi\} \cup \Sub(\phi) \cup \Sub(\psi) & \Sub(\lnext \phi) &\colonequals \{\lnext \phi\} \cup \Sub(\phi)\\
  \Sub(\phi \luntil \psi ) &\colonequals \{\phi \luntil \psi\} \cup \Sub(\phi) \cup \Sub(\psi) & \Sub(\jbox{t}_\agent\phi) &\colonequals \{\jbox{t}_\agent\phi\} \cup \Sub(\phi).
 \end{align*}

 \section{Axioms}
 \label{sect:Axioms}

  The axiom system for temporal justification logic consists of three parts, namely propositional logic, temporal logic, and justification logic.

 \subsection*{Propositional Logic}
 For propositional logic, we take
 \begin{enumerate}
  \setcounter{enumi}{\theenumsave}
  \item all propositional tautologies \hfill \propax
  \setcounter{enumsave}{\theenumi}
 \end{enumerate}
 as axioms and the rule modus ponens, as usual:
 \[
   \lrule{\phi \quad \phi \limplies \psi}{\psi}\,\mprule \, .
 \]

 \subsection*{Temporal Logic}
 For the temporal part, we use a system of~\cite{Gabbay,Gol87,Gor99} 
 with axioms
 \begin{enumerate}
  \setcounter{enumi}{\theenumsave}
  \item $\lnext( \phi \limplies \psi) \limplies (\lnext \phi \limplies \lnext \psi)$ \hfill \nextkax
  \item $\lalways( \phi \limplies \psi) \limplies (\lalways \phi \limplies \lalways \psi)$ \hfill \alwayskax
  \item $\lnext \lneg \phi \liff \lneg \lnext \phi$ \hfill \funax
  \item $\lalways (\phi \limplies \lnext \phi) \limplies (\phi \limplies \lalways \phi)$ \hfill \indax
  \item $\phi \luntil \psi \limplies \leventually \psi$ \hfill \uoneax
  \item $\phi \luntil \psi \liff \psi \lor (\phi \land \lnext(\phi \luntil \psi))$ \hfill \utwoax
  \setcounter{enumsave}{\theenumi}
 \end{enumerate}
 and rules
 \[
  \lrule{\phi}{\lnext \phi}\,\nextnecrule \, , \qquad\qquad \lrule{\phi}{\lalways\phi}\,\alwaysnecrule \, .
 \]

 \subsection*{Justification Logic}
 Finally, for the justification logic part, we use a multi-agent version of the Logic of Proofs~\cite{Art01BSL,BucKuzStu11JANCL,Ghari:2014,TYav08TOCS} with axioms
 \begin{enumerate}
  \setcounter{enumi}{\theenumsave}
  \item $\jbox{t}_\agent (\phi \limplies \psi) \limplies (\jbox{s}_\agent \phi \limplies \jbox{t \tapp s}_\agent \psi)$ \hfill \appax
  \item $\jbox{t}_\agent \phi \rightarrow \jbox{t \tsum s}_\agent \phi$ \quad $\jbox{s}_\agent \phi \limplies  \jbox{t \tsum s}_\agent \phi$ \hfill \sumax
  \item $\jbox{t}_\agent \phi \limplies \phi$ \hfill \refax
  \item $\jbox{t}_\agent \phi \limplies \jbox{\tinspect t}_\agent \jbox{t}_\agent \phi$ \hfill \posintax
  \setcounter{enumsave}{\theenumi}
 \end{enumerate}
 and rule
 \[
   \lrule{\jbox{c}_\agent \phi \in \CS}{\jbox{c}_\agent \phi}\, \constnecrule \, ,
 \]
 where the constant specification $\CS$ is a set of formulas $\jbox{c}_\agent \phi$, where $c \in \CTerms$ is a justification constant and $\phi$ is an axiom of propositional logic, temporal logic, or justification logic.

For a given constant specification $\CS$, we use $\LPLTL_\CS$ to denote the Hilbert system given by the axioms and rules for propositional logic, temporal logic, and justification logic as presented above.
As usual, we write $\LPLTL_\CS \vdash \phi$ or simply $\vdash_\CS \phi$ if a formula $\phi$ is derivable in 
$\LPLTL_\CS$. 
Often the constant specification is clear from the context and we will only write $\vdash \phi$ instead of $\vdash_\CS \phi$.

The axiomatization for linear time temporal logic given in~\cite{Gabbay,Gol87,Gor99} includes an axiom 
\[
\Box \phi \to (\phi \land \lnext \Box \phi).
\]
The following lemma shows that we do not need this axiom since in our formalization $\Box$ is a defined operator. 

\begin{lemma}\label{l:mix:1}
We have
\[
\vdash_\CS \Box \phi \to (\phi \land \lnext \Box \phi)
\]
and $\mprule$ is the only rule that is used in this derivation.
\end{lemma}
\begin{proof}
$\Box \phi$ stands for $\lnot ( \top \luntil \lnot\phi)$.
Hence from $\utwoax$ we get
\[
\vdash_\CS  \lnot \phi \lor \lnext (\top \luntil \lnot \phi) \to \top \luntil \lnot \phi.
\]
Taking the contrapositive yields
\[
\vdash_\CS  \lnot ( \top \luntil \lnot \phi) \to \lnot (\lnot \phi \lor \lnext (\top \luntil \lnot \phi)).
\]
By propositional reasoning and $\funax$ we get
\[
\vdash_\CS  \lnot ( \top \luntil \lnot \phi) \to  (\phi \land \lnext \lnot (\top \luntil \lnot \phi)),
\]
which is
\[
\vdash_\CS  \Box \phi \to (\phi \land \lnext \Box \phi).
\qedhere
\]
\end{proof}

\begin{remark}
As usual, we find that the following rule is derivable, see~\cite[Lemma~6]{DBLP:journals/corr/Bucheli15} for a detailed derivation,
\[  
\lrule{\chi \limplies \lneg \psi \land \lnext \chi}{\chi \limplies \lneg(\phi \luntil \psi)}\, . 
\]
From this, we get that the following rule  is also derivable
\[  
\lrule{\chi \limplies \lneg \psi \land \lnext (\chi \lor (\lnot \phi \land \lnot \psi)) }{\chi \limplies \lneg(\phi \luntil \psi)}\,\uindrule \, . 
\]
A proof is given in~\cite[Lemma~4.5]{HvdMV04}. 
\end{remark}


 \section{Semantics}
 \label{sect:Semantics}

In this section we introduce interpreted systems based on  Fitting-models as semantics for temporal justification logic.

\begin{definition}
A \emph{frame} is a tuple $(S, R_1,\ldots,R_\numberofagents)$ where
\begin{enumerate}
\item $S$ is a non-empty set of states;
\item each $R_i \subseteq S \times S$ is a reflexive and transitive relation. 
\end{enumerate}
A \emph{run}~$r$ on a frame is a function from $\N$ to states, i.e., $r: \N \to S$. A \emph{system}~$\runs$ is a non-empty set of runs. 
\end{definition}

\begin{definition}\label{def:evidence function for LPLTL}
Given a frame $(S, R_1,\ldots,R_\numberofagents)$,
a \emph{$\CS$-evidence function for agent~$\agent$} is a function 
\[
\evidence_i: S \times \Terms \to \powerset(\Formulae)
\]
satisfying the following conditions.
For all terms $s,t \in \Terms$, all formulas $\phi,\psi \in \Formulae$, and all $v,w \in S$,
 \begin{enumerate}
  \item 
	$\evidence_\agent(v,t) \subseteq \evidence_\agent(w, t)$, whenever $R_i(v,w)$\hfill (monotonicity)
  \item 
	if $\jbox{c}_\agent \phi \in \CS$, then $\phi \in \evidence_\agent(w,c)$ \hfill (constant specification)
\item 
	if $\phi \limplies \psi \in \evidence_\agent(w,t)$ and $\phi \in \evidence_\agent(w,s)$, then $\psi \in \evidence_\agent(w, t \tapp s)$ \hfill (application)
  \item 
	$\evidence_\agent(w,s) \cup \evidence_\agent(w,t) \subseteq \evidence_\agent(w,s \tsum t)$ \hfill (sum)
  \item 
	if $\phi \in \evidence_\agent(w,t)$, then $\jbox{t}_\agent \phi \in \evidence_\agent(w,\tinspect t)$ \hfill (positive introspection)
 \end{enumerate}
\end{definition}

\begin{definition}
An  \emph{interpreted system for $\CS$}\/ is a tuple 
\[
\system = (\runs, S, R_1,\ldots,R_\numberofagents, \evidence_1\ldots,\evidence_\numberofagents, \valuation)
\] 
where
\begin{enumerate}
\item $(S, R_1,\ldots,R_\numberofagents)$ is a frame;
\item $\runs$ is a system on that frame;
\item $\evidence_\agent$ is a $\CS$-evidence function for agent~$\agent$ for $1 \leq \agent \leq \numberofagents$;
\item $\valuation: S \to \powerset(\Prop)$ is a valuation.
\end{enumerate}
\end{definition}

\begin{definition}
Given an interpreted system 
\[
\system = (\runs, S, R_1,\ldots,R_\numberofagents, \evidence_1,\ldots,\evidence_\numberofagents, \valuation),
\] 
a run $r \in \runs$, and  $n \in \N$, we define truth of a formula $\phi$ in $\system$ at state $r(n)$ inductively by 
 \begin{align*}
  (\system, r, n) &\entails P \text{ iff } P \in \valuation(r(n)) \, ,\\
  (\system, r, n) &\not\entails \lfalse \, ,\\
  (\system, r, n) &\entails \phi \limplies \psi \text{ iff } (\system, r, n) \not\entails \phi \text{ or } (\system, r, n) \entails \psi \, ,\\
  (\system, r, n) &\entails \lnext \phi \text{ iff } (\system, r, n+1) \entails \phi \, ,\\
   (\system, r, n) &\entails \phi \luntil \psi \text{ iff there is some } m \geq 0 \text{ such that } (\system, r, n+m) \entails \psi \\ & \qquad\qquad \text{ and } (\system, r, n+k) \entails \phi \text{ for all } 0 \leq k < m \, ,\\    
  (\system, r, n) &\entails \jbox{t}_\agent \phi \text{ iff }  \phi \in \evidence_\agent(r(n),t)  \text { and } (\system, r^\prime, n^\prime) \entails \phi \\ &\qquad\qquad \text{ for all } r^\prime \in \runs \text{ and } n^\prime \in \N \text{ such that } R_\agent(r(n) , r^\prime(n^\prime)) \, .
 \end{align*}

 As usual, we write $\system \entails \phi$ if
for all $r \in \runs$ and all $ n \in \N$, we have 
$(\system, r, n) \entails \phi$.
Further, we write $\entails_\CS \phi$ if $\system \entails \phi$ for all 
  interpreted systems $\system$ for $\CS$.
\end{definition}
 
\begin{remark}
From the definitions of $\Box$ and $\Diamond$ it follows that:  \begin{align*}
 (\system, r, n) &\entails \Diamond \phi \text{ iff } (\system, r, n+k) \entails \phi \text{ for some } k\geq 0 \, ,\\
 (\system, r, n) &\entails \lalways \phi \text{ iff } (\system, r, n+k) \entails \phi \text{ for all } k\geq 0 \, .\\
 \end{align*}
\end{remark}

\section{Soundness and Completeness}
\label{sect:Completeness}

The soundness proof for $\LPLTL_\CS$ is a straightforward combination of the soundness proofs for temporal logic and justification logic by induction on the derivation.

\begin{theorem}
Let $\CS$ be an arbitrary constant specification.
For each formula $\phi$,
   \[
   \vdash_\CS \phi  \quad\text{implies}\quad \models_\CS \phi.
   \]
\end{theorem}
Our completeness proof for $\LPLTL_\CS$ follows the one given in~\cite{HvdMV04}. 
First, we define
\[
\Gamma \vdash_\CS \phi \text{ iff there exist $\psi_1,\ldots,\psi_n \in \Gamma$ such that 
$\vdash_\CS (\psi_1 \land \cdots \land \psi_n) \to \phi$.}
\]
Following our convention, we will usually write $\Gamma \vdash \phi$ instead of $\Gamma \vdash_\CS \phi$.

\begin{definition}
 Let $\CS$ be a constant specification. A set $\Gamma$ of formulas is called \emph{$\CS$-consistent} if $\Gamma \not\vdash_\CS \bot$.
That means
$\not\vdash_\CS \bigwedge\Sigma \rightarrow\bot$, for each finite~$\Sigma \subseteq \Gamma$.
\end{definition}

For a formula $\chi$, let 
$\Sub^+(\chi) := \Sub(\chi) \cup \{ \neg \psi \ |\  \psi \in \Sub(\chi) \}$.
Let $\MCS_\chi$ denote the set of all maximally $\CS$-consistent subsets of $\Sub^+(\chi)$. We have the following facts for $\Gamma \in \MCS_\chi$:
 \begin{itemize}
 \item If 
$\Gamma \vdash_\CS \phi$, then $\vdash_\CS \bigwedge \Gamma \to \phi$.
 
\item If $\phi \in \Sub(\chi)$ and  $\phi \not\in \Gamma$, then $\neg \phi \in \Gamma$.

\item If $\phi \in \Sub^+(\chi)$ and  $\Gamma \vdash_\CS \phi$, then $\phi \in \Gamma$.

\item If $\psi \in \Sub^+(\chi)$, $\phi \in \Gamma$ and $\vdash_\CS \phi \rightarrow \psi$, then $\psi \in \Gamma$.
\end{itemize} 
    
We define the relation $R_\lnext$ on $\MCS_\chi$ as follows:
\[
\Gamma R_\lnext \Delta 
\quad\text{if{f}}\quad 
\nvdash_\CS \bigwedge \Gamma \to \lnot \lnext \bigwedge \Delta.
\]
From this definition we immediately get the following lemmas.

\begin{lemma}\label{lem:R-next}
Let $\Gamma,\Delta \in \MCS_\chi$, $\Gamma R_\lnext \Delta$, and $\phi \in \Sub(\chi)$.

\begin{enumerate}
\item  If\/ $\Gamma \vdash_\CS \lnext \phi$, then $\phi \in \Delta$.

\item  If\/ $\Gamma \vdash_\CS \neg \lnext \phi$, then $\neg \phi \in \Delta$.
\end{enumerate}
\end{lemma}
\begin{proof}
\begin{enumerate}
\item 
Suppose toward a contradiction that $\phi \not \in \Delta$. Thus $\neg \phi \in \Delta$. Since $\Gamma \vdash_\CS \lnext \phi$, we have
\mbox{$
\vdash_\CS \bigwedge\Gamma \to \lnext \phi.
$}
Hence 
$
\vdash_\CS \bigwedge\Gamma \to \lnext \neg\neg \phi.
$
 Therefore 
 $\vdash_\CS \bigwedge\Gamma \to \lnext \neg \bigwedge \Delta.
 $
  Thus 
  \[
  \vdash_\CS \bigwedge\Gamma \to \neg \lnext  \bigwedge \Delta,
  \]
   which would contradict $\Gamma R_\lnext \Delta$.
   
\item The proof is similar to part 1. \qedhere   
\end{enumerate}
\end{proof}

\begin{lemma}\label{l:4.4}
Let $\Gamma \in \MCS_\chi$ and let 
$S := \{ \Delta \in  \MCS_\chi \ |\  \Gamma R_\lnext \Delta\}$.
We have
\[
\vdash \bigwedge \Gamma \to \lnext \bigvee \big\{ \bigwedge \Delta \ |\ \Delta \in S \big\}.
\]
\end{lemma}
\begin{proof}
First observe that for all $\Gamma, \Delta \in \MCS_\chi$ we have
\begin{equation}\label{eq:next:1}
(\text{not } \Gamma R_\lnext \Delta) 
\quad\text{implies}\quad
\vdash \bigwedge \Gamma \to \lnot \lnext \bigwedge \Delta.
\end{equation}
We also have
\[
\vdash \bigvee \big\{ \bigwedge \Delta \ |\ \Delta \in  \MCS_\chi \big\}.
\]
By necessitation we get
\[
\vdash \lnext\bigvee \big\{ \bigwedge \Delta \ |\ \Delta \in  \MCS_\chi \big\}
\]
and thus
\begin{equation}\label{eq:next:2}
\vdash \bigvee \big\{ \lnext \bigwedge \Delta \ |\ \Delta \in  \MCS_\chi \big\}.
\end{equation}
By \eqref{eq:next:1} we infer
\[
\vdash \bigwedge \Gamma \to \bigvee \big\{ \lnext \bigwedge \Delta \ |\ \Delta \in  \MCS_\chi  \text{ with } \Gamma R_\lnext \Delta \big\}
\]
and thus
\[
\vdash \bigwedge \Gamma \to \lnext \bigvee \big\{\bigwedge \Delta \ |\ \Delta \in  \MCS_\chi  \text{ with } \Gamma R_\lnext \Delta \big\}.
\qedhere
\]
\end{proof}

\begin{lemma}\label{l:serial:1}
The relation $R_\lnext$ is serial. 
That is for each $\Gamma \in \MCS_\chi$, there exists $\Delta \in \MCS_\chi$ with $\Gamma R_\lnext \Delta$.
\end{lemma}
\begin{proof}
Suppose towards a contradiction that for $\Gamma \in \MCS_\chi$ we have (not $\Gamma R_\lnext \Delta$) for all $\Delta \in \MCS_\chi$. Then 
$
\vdash \bigwedge \Gamma \to \lnot \lnext \bigwedge \Delta,
$
for all $\Delta \in \MCS_\chi$. Thus
\[
\vdash \bigwedge \Gamma \to \bigwedge \{\lnot \lnext \bigwedge \Delta \ |\ \Delta \in \MCS_\chi  \big\},
\]
and hence,
\begin{equation}\label{eq:serial:1}
\vdash \bigwedge \Gamma \to \lnot \bigvee \{ \lnext \bigwedge \Delta \ |\ \Delta \in \MCS_\chi  \big\}.
\end{equation}
On the other hand, from \eqref{eq:next:2} we deduce
\begin{equation}\label{eq:serial:2}
\vdash \bigwedge \Gamma \to \bigvee \big\{ \lnext \bigwedge \Delta \ |\ \Delta \in  \MCS_\chi \big\}.
\end{equation}
Since $\Gamma$ is consistent, \eqref{eq:serial:1} and \eqref{eq:serial:2} leads to a contradiction.
\end{proof}

 \begin{definition}
A finite sequence $(\Gamma_0, \Gamma_1, \ldots, \Gamma_n)$ of elements of~$\MCS_\chi$ is called a \emph{$\phi \luntil \psi$-sequence starting with $\Gamma$} if
\begin{enumerate}
\item $\Gamma_0 = \Gamma$,

\item $\Gamma_j R_\lnext \Gamma_{j+1}$, for all $j < n$,

\item $\psi \in \Gamma_n$,

\item $\phi \in \Gamma_j$, for all $j < n$.
\end{enumerate}
 \end{definition}

\begin{lemma}\label{lem:until-sequence}
For every $\Gamma \in \MCS_\chi$, if $\phi \luntil \psi \in \Gamma$, then there exists a  $\phi \luntil \psi$-sequence starting with $\Gamma$.
\end{lemma}
\begin{proof}
Suppose $\phi \luntil \psi \in \Gamma$ and there exists no  $\phi \luntil \psi$-sequence starting with~$\Gamma$. 
We let $T$ be the smallest set of elements of $\MCS_\chi$ such that
\begin{enumerate}
\item $\Gamma \in T$;
\item for each $\Delta' \in \MCS_\chi$, if $\Delta \in T$, $\Delta R_\lnext \Delta '$, and $\phi \in \Delta'$, then $\Delta' \in T$.
\end{enumerate}
We find that $\vdash \bigwedge \Delta \to \lnot \psi$ 
for all $\Delta \in T$. 
 Let 
\[
\rho := \bigvee \big\{ \bigwedge \Delta \ |\ \Delta \in T \big\}.
\]
We have $\vdash \rho \to \lnot \psi$. 

Moreover, for each $\Delta \in T$ and each $\Delta' \in \MCS_\chi$ with $\Delta R_\lnext \Delta '$ , we have 
\[
\text{either}\quad
\Delta' \in T \quad \text{or} \quad \vdash \bigwedge \Delta' \to \lnot \phi \land \lnot \psi.
\]
Thus, by Lemma~\ref{l:4.4},
we get 
$
\vdash \rho \to \lnext (\rho \lor (\lnot \phi \land \lnot \psi)).
$
Using $\uindrule$, we obtain $\vdash \rho \to \lnot(\phi \luntil \psi)$.
Since $\Gamma \in T$, this implies  
$
\vdash \bigwedge \Gamma \to \lnot(\phi \luntil \psi),
$
which contradicts the assumption
$\phi \luntil \psi \in \Gamma$.
\end{proof}
 
 \begin{definition}\label{Def:acceptable sequence}
An infinite sequence $(\Gamma_0, \Gamma_1, \ldots)$ of elements of $\MCS_\chi$ is called \textit{acceptable} if
\begin{enumerate}
\item $\Gamma_n R_\lnext \Gamma_{n+1}$ for all $n \geq 0$, and
\item for all $n$, if $\phi \luntil \psi \in \Gamma_n$, then there exists $m \geq n$ such that $\psi \in \Gamma_m$ and $\phi \in \Gamma_k$ for all $k$ with $n \leq k <m$.
\end{enumerate}
 \end{definition}
 
\begin{lemma}\label{lem:finite seq to acceptable seq}
Every finite sequence $(\Gamma_0, \Gamma_1, \ldots, \Gamma_n)$ of elements of\/ $\MCS_\chi$ with $\Gamma_j R_\lnext \Gamma_{j+1}$, for all $j < n$, can be extended to an  acceptable sequence.
\end{lemma}
\begin{proof}
In order to fulfill the requirements of Definition \ref{Def:acceptable sequence}, we shall extend the sequence $(\Gamma_0, \Gamma_1, \ldots, \Gamma_n)$ by the following algorithm.

Suppose $\phi \luntil \psi \in \Gamma_0$. Then either $\psi \in \Gamma_0$ or $\neg\psi \in \Gamma_0$. In the former case the requirement is fulfilled for the formula $\phi \luntil \psi$  in $\Gamma_0$, and we go to the next step. In the latter case, using axiom $(\luntil 2)$, 
\[
\Gamma_0 \vdash_\CS \phi \wedge \lnext (\phi \luntil \psi). 
\]
Since $\Gamma_0 R_\lnext \Gamma_1$, by Lemma \ref{lem:R-next}, we get $\phi \luntil \psi \in \Gamma_1$.  

We can repeat this argument for $\Gamma_i$ for 
$1\leq i \leq n$. We find that the requirement for $\phi \luntil \psi \in \Gamma_0$ is either fulfilled in $(\Gamma_0, \Gamma_1, \ldots, \Gamma_n)$ or  $\phi \luntil \psi \in \Gamma_n$ and  $\phi \in \Gamma_i$ for $1\leq i \leq n$. 
In the latter case, by Lemma~\ref{lem:until-sequence}, there exists a sequence $(\Gamma_n, \Gamma_{n+1}, \ldots, \Gamma_{n+m})$
such that 
$\phi \in \Gamma_i$ for $n\leq i < n+m$,   
$\psi \in \Gamma_{n+m}$, 
and  $\Gamma_i R_\lnext \Gamma_{i+1}$ for $n\leq i < n+m$. This gives a finite extension of the original sequence that satisfies the requirement imposed by $\phi \luntil \psi \in \Gamma_0$.

In the next step we repeat this argument for the remaining obligations at $\Gamma_0$. Eventually we obtain a finite sequence that satisfies all requirements imposed by formulas at $\Gamma_0$.

We may move on to $\Gamma_1$ and apply the same procedure. 
It is clear that by iterating it we obtain in the limit an acceptable sequence that extends $(\Gamma_0, \Gamma_1, \ldots, \Gamma_n)$.
\end{proof}
 
\begin{corollary}\label{cor:acceptable sequence starts with Gamma}
For every $\Gamma \in \MCS_\chi$, there is an acceptable sequence that starts with $\Gamma$. 
\end{corollary}

\begin{definition}
The $\chi$-canonical interpreted system 
\[
\system = (\runs, S, R_1,\ldots,R_\numberofagents, \evidence_1\ldots,\evidence_\numberofagents, \valuation)
\] 
for $\CS$
is defined as follows:
\begin{enumerate}
\item $\runs$ consists of all mappings $r: \N \to \MCS_\chi$ such that
$(r(0), r(1), \ldots)$ is an acceptable sequence;
\item $S := \MCS_\chi = \{ r(n) \ |\ r\in \runs, n \in \N \}$;
\item $R_i(\Gamma, \Delta)$ if{f} 
$\{ \phi \ |\  \Gamma \vdash \jbox{t}_i \phi   \text{ for some $t$}\} \subseteq \{ \phi \ |\ \Delta \vdash \phi\}$; 
\item $\evidence_i (\Gamma, t) := \{ \phi \ |\  \Gamma \vdash \jbox{t}_i \phi \}$;
\item $\valuation(\Gamma) := \{P \in \Prop\ |\  P \in \Gamma\}$.
\end{enumerate}
\end{definition}

\begin{remark}
The $\chi$-canonical interpreted system~$\system$ for~$\CS$ is a finite structure in the sense that the set of states~$S$ is finite.
This is a novelty for completeness proofs of justification logics. Even the completeness proofs for justification logics with common knowledge~\cite{Art06TCS,BucKuzStu11JANCL} work with infinite canonical structures. Note that this remark concerns epistemic Fitting-models. Of course, symbolic M-models~\cite{Mkr97LFCS} could be considered as single-world Fitting-models.

The fact that states of~$\system$ are maximally $\CS$-consistent subsets of $\Sub^+(\chi)$---instead of just maximally $\CS$-consistent sets---%
matters
for the definitions of $R_i$ and $\evidence_i$.
The usual definitions would be 
\begin{gather*}
R_i(\Gamma, \Delta)  \text{ if{f} } 
\{ \phi \ |\  \jbox{t}_i \phi \in \Gamma   \text{ for some $t$}\} \subseteq \{ \phi \ |\  \phi \in \Delta\} \qquad \text{and }\\
\evidence_i (\Gamma, t) := \{ \phi \ |\  \jbox{t}_i \phi \in \Gamma \}.
\end{gather*}
This, however, would not work for our finite canonical structure. In particular the next lemma could not be established as, for instance, 
$\jbox{t}_i \phi \in \Gamma$ does not imply $\jbox{!t}_i  \jbox{t}_i \phi \in \Gamma$ for $\Gamma \in \MCS_\chi$.
\end{remark}

\begin{lemma}
The $\chi$-canonical interpreted system 
\[
\system = (\runs, S, R_1,\ldots,R_\numberofagents, \evidence_1\ldots,\evidence_\numberofagents, \valuation)
\] 
for $\CS$ is an interpreted system for $\CS$.
\end{lemma}
\begin{proof}
The proof is essentially the same as the corresponding proof for single agent Fitting-models in~\cite{Fit05APAL}.
Let us only show here the monotonicity condition for $\evidence_i$.

Suppose $\Gamma, \Delta \in S$ and $R_i(\Gamma, \Delta)$. 
Suppose that $\phi \in \evidence_i (\Gamma,t)$. Thus $\Gamma \vdash \jbox{t}_i \phi$. Hence $\Gamma \vdash \jbox{!t}_i \jbox{t}_i \phi$. Since $R_i(\Gamma,\Delta)$, we have $\Delta \vdash \jbox{t}_i \phi$. Hence $\phi \in \evidence_i (\Delta,t)$ as desired.
\end{proof}

\begin{lemma}[Truth Lemma]
Let 
$
\system = (\runs, S, R_1,\ldots,R_\numberofagents, \evidence_1\ldots,\evidence_\numberofagents, \valuation)
$ 
be the $\chi$-canonical interpreted system for~$\CS$. For every formula $\psi \in \Sub^+(\chi)$, every run~$r$ in $\runs$, and every $n \in \N$ we have:
\[
(\system, r, n) \models \psi 
\quad\text{if{f}}\quad
 \psi \in r(n). 
\]
\end{lemma}
\begin{proof}
As usual, the proof is by induction on the structure of $\psi$. We show only the following cases:
\begin{itemize}
\item $\psi = \jbox{t}_i \phi$.
$(\Rightarrow)$ If $(\system, r, n) \models \jbox{t}_i \phi$, then $\phi \in \evidence_i (r(n),t)$. Thus, by definition, $r(n) \vdash \jbox{t}_i \phi$. Hence $\jbox{t}_i \phi \in r(n)$, since $\jbox{t}_i \phi \in \Sub^+(\chi)$. 

$(\Leftarrow)$ If $\jbox{t}_i \phi \in r(n)$, then $r(n) \vdash \jbox{t}_i \phi$. Hence, by definition, $\phi \in \evidence_i (r(n),t)$. 
Now suppose that $R_i (r(n), r'(n'))$. We find $r'(n') \vdash \phi$. Since $\phi  \in \Sub^+(\chi)$, we have $\phi \in r'(n')$ and by I.H.~we get $(\system, r', n') \models \phi$. Since $r'$ and $n'$ were arbitrary, we conclude $(\system, r, n) \models \jbox{t}_i \phi$.
\item $\psi = \lnext \phi$.
$(\Rightarrow)$ Suppose that $(\system, r, n) \models \lnext \phi$ and $\lnext \phi \not\in r(n)$.  Then $(\system, r, n+1) \models  \phi$, and hence by the induction hypothesis $\phi \in r(n+1)$. On the other hand, $\neg \lnext \phi \in r(n)$. Since $r(n) R_\lnext r(n+1)$, by Lemma \ref{lem:R-next}, we get  $\neg  \phi \in r(n+1)$, which is a contradiction.

$(\Leftarrow)$ If $\lnext \phi \in r(n)$, then $\phi \in r(n+1)$. By the induction hypothesis, $(\system, r, n+1) \models  \phi$, and hence $(\system, r, n) \models \lnext \phi$.
\item $\psi = \psi_1 \luntil \psi_2$.
$(\Rightarrow)$ If $(\system, r, n) \models  \psi_1 \luntil \psi_2$, then $(\system, r, m) \models  \psi_2$ for some $m \geq n$, and $(\system, r, k) \models  \psi_1$ for all $k$ with $n \leq k < m$. By I.H.~we get $\psi_2 \in r(m)$, and $\psi_1 \in r(k)$ for all~$k$ with $n \leq k < m$. We have to show $\psi_1 \luntil \psi_2 \in r(n)$, which follows by induction on $m$ as follows:
\begin{itemize}
\item
Base case $m=n$. Since $\psi_2 \in r(n)$ and $\vdash \psi_2 \rightarrow (\psi_1 \luntil \psi_2)$, we obtain $\psi_1 \luntil \psi_2 \in r(n)$. 
\item
Suppose $m > n$. It follows from the induction hypothesis that $\psi_1 \luntil \psi_2 \in r(n+1)$. 
From this and $r(n) R_\lnext r(n+1)$ we get that
\begin{equation}\label{eq:tl:1}
\text{$r(n) \cup \{\lnext(\psi_1 \luntil \psi_2) \}$ is consistent.}
\end{equation}
Assume now
\begin{equation}\label{eq:tl:2}
\lnot(\psi_1 \luntil \psi_2) \in r(n).
\end{equation}
Then
$
r(n) \vdash \lnot(\psi_1 \luntil \psi_2)
$
and by axiom~$\utwoax$ we find
$
r(n) \vdash \lnot(\psi_1 \land \lnext(\psi_1 \luntil \psi_2))
$.
From $\psi_1 \in r(n)$ we get $r(n) \vdash \psi_1$ and thus
$
r(n) \vdash \lnot \lnext(\psi_1 \luntil \psi_2),
$
which contradicts~\eqref{eq:tl:1}.
Hence the assumption~\eqref{eq:tl:2} must be false and we conclude
$\psi_1 \luntil \psi_2 \in r(n)$.
\end{itemize}

$(\Leftarrow)$ If $\psi_1 \luntil \psi_2 \in r(n)$, then since $(r(n),r(n+1),\ldots)$ is an acceptable sequence there exists $m \geq n$ such that $\psi_2 \in r(m)$, and $\psi_1 \in r(k)$ for all~$k$ with $n \leq k < m$. By
I.H.~we obtain 
$(\system, r, m) \models  \psi_2$, and $(\system, r, k) \models  \psi_1$ for all $k$ with $n \leq k < m$. Thus 
$(\system, r, n) \models  \psi_1 \luntil \psi_2$.
\qedhere
\end{itemize}
\end{proof}

\begin{theorem}[Completeness]\label{thm:Completeness-interpreted systems}
For each formula $\phi$,
   \[
   \models_\CS \phi  \quad\text{implies}\quad \vdash_\CS \phi.
   \]
\end{theorem}
\begin{proof}
Suppose that $\not\vdash_\CS \phi$. Thus, $\{ \neg \phi\}$ is a $\CS$-consistent set.
Therefore, there exists $\Gamma \in \MCS_{\phi}$ with $\neg \phi \in \Gamma$.
By Corollary \ref{cor:acceptable sequence starts with Gamma}, there is an acceptable sequence starting with $\Gamma$.
Thus there is a run $r$ in the $\phi$-canonical interpreted system~$\system$ for $\CS$ with $r(0)=\Gamma$.
Since $\neg \phi \in \Gamma$, by the Truth Lemma, $(\system, r, 0) \not\models \phi$. Therefore, $\not\models_\CS \phi$.
\end{proof}

\section{Internalization}\label{sect:internalization}

It is desirable that a justification logic internalizes its own notion of proof. This is formalized in the following definition.

\begin{definition}
A justification logic $\mathsf{L}$ satisfies \emph{internalization} if for each formula $\phi$ with
$
\mathsf{L} \vdash \phi
$
and for each agent $\agent$, there exists a term $t$ with
$
\mathsf{L} \vdash \jbox{t}_\agent \phi 
$.
\end{definition}

Usually, internalization is shown by induction on the derivation of $\phi$. However, for $\LPLTL_\CS$ this seems not possible because it includes rules $\nextnecrule$ and $\alwaysnecrule$.
In this section, we introduce an extension $\LL_\CS$ of $\LPLTL_\CS$ that satisfies internalization.

The language of $\LL_\CS$ includes a new unary operator $\star$ on justification terms. We define
\[
\star^0 c  \colonequals c \qquad\text{and}\qquad \star^n c  \colonequals \star \star^{n-1}  c \quad\text{(for $n\geq 1$)}\,.
\]
The set of terms $\Terms^\star$ of  $\LL_\CS$ is  given by
\[
  t \coloncolonequals \star^n c \mid x \mid \; \tinspect t \mid \;  t \tsum t \mid t \tapp t \, ,
 \]
 where $c \in \CTerms$, $n \geq 0$, and $x \in \VTerms$. 
The set of formulas $\Formulae^\star$ of  $\LL_\CS$ is defined like $\Formulae$ but using $\Terms^\star$ instead of $\Terms$.

The axioms of $\LL_\CS$ are:
\begin{enumerate}
\item all axioms of $\LPLTL$
\item $\Box \phi \to \lnext \phi  \hfill \mixprinciple$
\item $\Box (\jbox{t}_{i} \phi \to \phi) \hfill \boxrefax$ 
\end{enumerate}
The rules of $\LL_\CS$ are:
\[
  \lrule{\phi \quad \phi \limplies \psi}{\psi}\,\mprule 
\qquad
\text{and}
 \qquad
\lrule{\jbox{c}_{i_0} \phi \in \CS}{\jbox{\star^n c}_{i_n}\Box \jbox{\star^{n-1} c}_{i_{n-1}}\Box\ldots \jbox{\star c}_{i_1} \Box \jbox{c}_{i_0} \phi}\,  \constnecrule^\star \,,
\]
where $n \geq 0$; so  $\constnecrule^\star$ subsumes $\constnecrule$.
Note that a constant specification for $\LL$ may include formulas of the form 
$
\jbox{c}_i (\Box \phi \to \lnext \phi)$
and 
$
\jbox{c}_i \Box (\jbox{t}_{i} \phi \to \phi).
$

\begin{remark}
The principles $\mixprinciple$ and $\boxrefax$ are derivable in $\LPLTL_\CS$.
However, their proofs require applications of the rules $\nextnecrule$ and $\alwaysnecrule$, respectively.
Since these rules are not included in $\LL_\CS$, we have to include $\mixprinciple$ and $\boxrefax$ as axioms.
\end{remark}

\begin{remark}
The $\star$-operation is very powerful. Its meaning can be explained as follows.
If $\jbox{c}_i \phi$ is contained in $\CS$, then $\jbox{c}_i \phi$ is provable and hence $\Box \jbox{c}_i \phi$ is provable, too (see Lemma~\ref{l:nexDeriv:1}). The evidence $\star c$ justifies this fact, i.e., $\jbox{\star c}_i \Box \jbox{c}_i \phi$ is provable.
Looking closely at $\constnecrule^\star$ we see that we get even more. Indeed, for any agent~$j$ we have that
$\jbox{\star c}_j \Box \jbox{c}_i \phi$ is provable. Moreover, even arbitrary iterations of this principle are provable, which implies that the constant specification is common knowledge among the agents, so to speak.

We could use a less general version of $\constnecrule^\star$ where the $\star$-operation is indexed. This would be similar to the evidence verification operation of~\cite{TYav08TOCS}, see also Question~\ref{q:yav}.  In that case we would obtain $\jbox{\star^j_i c}_j \Box \jbox{c}_i \phi$.
However, for the purpose of internalization we do not need these indices and hence we dispense with them.
\end{remark}

\begin{definition}
A constant specification $\CS$ is \emph{axiomatically appropriate} if for each axiom~$\phi$ of $\LL$ and each agent $i$, there is a constant $c$ with $\jbox{c}_i \phi \in \CS$.
\end{definition}


First we show that $\LL_\CS$ extends $\LPLTL_\CS$.

\begin{lemma}\label{l:nexDeriv:1}
Let $\CS$ be an axiomatically appropriate constant specification for $\LL$. The rules $\alwaysnecrule$ and $\nextnecrule$ are derivable in $\LL_\CS$.
\end{lemma}
\begin{proof}
We first show that $\alwaysnecrule$ is derivable in $\LL_\CS$. Suppose $\phi$ is provable in $\LL_\CS$. By induction on the proof of $\phi$, we show that $\Box\phi$ is provable in $\LL_\CS$. 

In case $\phi$ is an axiom, since $\CS$ is axiomatically appropriate, there is a constant $c$ such that $\jbox{c}_\agent \phi \in \CS$. Using $\constnecrule^\star$, we get $\jbox{\star c}_\agent \Box  \jbox{c}_\agent \phi$, and then using axiom $\refax$ we get $\Box  \jbox{c}_\agent \phi$. Finally, using axioms $\boxrefax$ and $(\Box\kax)$ we obtain $\Box \phi$.
 
In case $\phi$ is derived by modus ponens, the claim is immediate by $(\Box\kax)$.

In case $\phi$ is $\jbox{\star^n c}_{i_n}\Box \jbox{\star^{n-1} c}_{i_{n-1}}\Box\ldots \jbox{\star c}_{i_1} \Box \jbox{c}_{i_0} \phi$ derived using $\constnecrule^\star$, we can use $\constnecrule^\star$ also to obtain 
\[\jbox{\star^{n+1} c}_{i_{n+1}}\Box \jbox{\star^n c}_{i_n}\Box \jbox{\star^{n-1} c}_{i_{n-1}}\Box\ldots \jbox{\star c}_{i_1} \Box  \jbox{c}_{i_0} \phi.\] Then using $\refax$ we get 
\[\Box \jbox{\star^n c}_{i_n}\Box \jbox{\star^{n-1} c}_{i_{n-1}}\Box\ldots \jbox{\star c}_{i_1} \Box \jbox{c}_{i_0} \phi,\] that is $\Box \phi$.

Derivability of $\nextnecrule$ follows from $\alwaysnecrule$ and axiom $\mixprinciple$.
\end{proof}

Let $\CS$ be a constant specification for $\LL$. We set
\[
\CS^r \colonequals \{ \jbox{c}_i \phi \ |\  \jbox{c}_i \phi \in \CS \text{ and $\phi$ is an axiom of $\LPLTL$} \}\,.
\]
Obviously, $\CS^r$ is a constant specification for $\LPLTL$. We get the following corollary.

\begin{corollary}
Let $\CS$ be an axiomatically appropriate constant specification for $\LL$.
For each formula $\phi$ of  $\Formulae$,
\[
\LPLTL_{\CS^r} \vdash \phi \quad\text{implies} \quad \LL_\CS \vdash \phi.
\] 
\end{corollary}

We will now establish the internalization property. We need the following lemma.

\begin{lemma}\label{l:weakInt:1}
Let $\CS$ be an axiomatically appropriate constant specification.
For each formula~$\phi$ and each  $\agent$, 
\[
\LL_\emptyset \vdash \phi \quad\text{implies}\quad \LL_\CS \vdash  \jbox{t}_\agent \phi \text{ for some term $t$}. 
\]
\end{lemma}
\begin{proof}
We proceed by induction on the derivation of $\phi$.

In case $\phi$ is an axiom, since $\CS$ is axiomatically appropriate, there is a constant $c$ with
\[
\LL_\CS \vdash \jbox{c}_\agent \phi.
\]

In case $\phi$ is derived by modus ponens from $\psi \limplies \phi$ and $\psi$, then, by the induction hypothesis, there are term $s_1$ and $s_2$ such that $\jbox{s_1}_\agent (\psi \limplies \phi)$ and $\jbox{s_2}_\agent \psi$ are provable.
Using $\appax$ and modus ponens, we obtain $\jbox{s_1 \tapp s_2}_\agent \phi$.
\end{proof}

\begin{theorem}
Let $\CS$ be an axiomatically appropriate constant specification.
$\LL_\CS$ enjoys internalization.
\end{theorem}
\begin{proof}
We have to show that for each formula $\phi$ and each $\agent$
\[
\LL_\CS \vdash \phi
\quad\text{implies}\quad
\LL_\CS \vdash \jbox{t}_\agent \phi \text{ for some term $t$}.
\]
We proceed by induction on the derivation of $\phi$.

The cases where $\phi$ is an axiom or $\phi$ is derived by modus ponens are like the corresponding cases in the previous lemma.

In case $\phi$ is $\jbox{\star^n c}_{i_n}\Box \ldots \jbox{\star c}_{i_1} \Box \jbox{c}_{i_0} \psi$ derived using $\constnecrule^\star$, we can use $\constnecrule^\star$ also to obtain 
$\jbox{\star^{n+1} c}_{\agent}\Box \phi$. 
By Lemma~\ref{l:mix:1} we find $\LL_\emptyset \vdash \Box \phi \to \phi$.
Hence by Lemma~\ref{l:weakInt:1} there is a term $t$ such that
$\LL_\CS \vdash \jbox{t}_{\agent} (\Box \phi \to \phi)$.
We finally conclude
\[
\LL_\CS \vdash \jbox{t \tapp \star^{n+1} c }_{\agent} \phi \,.
\qedhere
\]
\end{proof}

It is straightforward to adapt our semantics for $\LPLTL_\CS$ to the extended language of $\LL_\CS$. Soundness and completeness of $\LL_\CS$ can then be shown similar to the case of $\LPLTL_\CS$. However, for the completeness proof of $\LL_\CS$ we require $\CS$ to be axiomatically appropriate in order to have the necessitation rules available.

\begin{definition}
	Let $\CS$ be a constant specification for $\LL$. A $\CS$-evidence function for agent $i$ for $\LL$ is a function 
	$
	\evidence_{i}: S \times \Terms^\star \to \powerset(\Formulae^\star)
	$
	satisfying conditions 1--5 of Definition \ref{def:evidence function for LPLTL} and the following additional condition:
	
	\begin{itemize}
		\item if $\jbox{c}_{i_0} \phi \in \CS$, 
		then 
		for all $w \in S$, all $n \geq 1$, and all agents $i_{n-1}, \ldots, i_1$:
		$$\Box \jbox{\star^{n-1} c}_{i_{n-1}} \ldots \Box\jbox{\star c}_{i_1} \Box \jbox{c}_{i_0} \phi \in \evidence_{i} (w,\star^n c).$$
	\end{itemize}
An $\LL_\CS$-interpreted system is an interpreted system where we use evidence functions for $\LL$.
We write $\models_\CS^\star \phi$ to mean $\system\entails \phi$ for all $\LL_\CS$-interpreted systems $\system$.
\end{definition}

\begin{theorem}[Soundness and completeness]\label{thm:completeness LPLTL^*}
	Let $\CS$ be an axiomatically appropriate constant specification for $\LL$. For each formula $\phi$,
	\[
	\models_\CS^\star \phi  \quad\text{if{f}}\quad  \LL_\CS \vdash \phi.
	\]
\end{theorem}

We conclude this section by showing the conservativity of $\LL$ over $\LPLTL$. First we need a lemma.

\begin{lemma}\label{lem:conservativity}
Let $\CS$ be a constant specification for $\LPLTL$, and $\system$ be an interpreted system of $\LPLTL$ for $\CS$. Then we can extend $\system$ to an $\LL_\CS$-interpreted system $\system^\star$ such that for every run $r$, every $n \in \mathbb{N}$, and every formula $\phi \in \Formulae$:
\[
(\system,r,n) \models \phi \quad \Longleftrightarrow \quad (\system^\star,r,n) \models \phi.
\]
\end{lemma}
\begin{proof}
Let $\system = (\runs,S,R_1,\ldots,R_h,\evidence_1,\ldots,\evidence_h,\nu)$ be an arbitrary interpreted system of $\LPLTL$ for $\CS$. 
By a least fixed point construction, we can easily extend the $\CS$-evidence functions~$\evidence_i$, for $1 \leq i \leq h$, to $\CS$-evidence functions $\evidence_i^\star$  such that 
\begin{enumerate}
\item
$\system^\star = (\runs,S,R_1,\ldots,R_h,\evidence_1^\star,\ldots,\evidence_h^\star,\nu)$ is an $\LL_\CS$-interpreted system and
\item
for each formula $\phi \in \Formulae$,  each run $r$ and each $n \in \mathbb{N}$:
\[
(\system,r,n) \models \phi \quad \Longleftrightarrow \quad (\system^\star,r,n) \models \phi\,. \qedhere
\]
\end{enumerate}
\end{proof}

\begin{theorem}[Conservativity]
Let $\CS$ be a constant specification for $\LPLTL$ and $\phi \in \Formulae$ a formula. If $\LL_\CS \vdash \phi$,  then $\vdash_\CS \phi$.
\end{theorem}
\begin{proof}
Suppose that $\not\vdash_\CS \phi$. Then, by Theorem \ref{thm:Completeness-interpreted systems}, we have $\not\models_\CS \phi$. Thus there exists an interpreted system $\system$ of $\LPLTL$ for $\CS$ and a state $r(n)$ such that $(\system,r,n) \not\models \phi$. Now, by Lemma \ref{lem:conservativity}, we find an $\LL_\CS$-interpreted system $\system^\star$ such that $(\system^\star,r,n) \not\models \phi$. Therefore, by Theorem \ref{thm:completeness LPLTL^*}, we have $\LL_\CS \not\vdash \phi$ as desired.
\end{proof}

\section{Additional Principles}\label{sect:addidtional}

In $\LPLTL_\CS$, epistemic and temporal properties do not interact. On the other hand in $\LL_\CS$, there are some interactions between time and knowledge, in axiom $\boxrefax$ and rule $\constnecrule^\star$. Here we propose some principles that create a connection between justifications and temporal modalities.
 We assume the language for terms to be augmented in the obvious way.
 \begin{gather*} 
  \jbox{t}_\agent \lalways \phi \limplies \lalways \jbox{\talwaysaccess t}_\agent \phi 
	\tag*{\alwaysaccessprinciple}\\
  \lalways \jbox{t}_\agent \phi \limplies \jbox{\tgeneralize t}_\agent \lalways  \phi
	\tag*{\generalizeprinciple}\\
  \jbox{t}_\agent \lnext \phi \limplies \lnext \jbox{\tnext t}_\agent \phi
	\tag*{\nextrightshiftprinciple}\\
  \lnext \jbox{t}_\agent \phi \limplies \jbox{\tprev t}_\agent \lnext \phi
	\tag*{\nextleftshiftprinciple}
 \end{gather*}
 Some first remarks about these principles:
  \begin{description}
   \item[\alwaysaccessprinciple] This is very plausible, if you have evidence that something always is true, then at every point in time you should be able to access this information. The term operator $\talwaysaccess$ makes the evidence accessible in every future point in time.
   \item[\generalizeprinciple] Using evidence this seems more plausible than just using knowledge, as one requires the evidence to be the same at every point in time. The term operator $\tgeneralize$ converts permanent evidence for a formula to evidence for believing that this formula is always true.
   \item[\nextrightshiftprinciple] This seems plausible: agents do not forget evidence once they have gathered it and can ``take it with them''. The term operator $\tnext$ carries evidence through time.
   \item[\nextleftshiftprinciple] This one seems less plausible as it implies some form of premonition. The term operator $\tprev$ presages future evidence for belief.
  \end{description}

The principle $\generalizeprinciple$ is very strong. In particular, it makes internalization possible even in the presence of necessitation rules.
Indeed, let  $\LPLTL^\mathsf{G}_\CS$ be the system $\LPLTL_\CS$ extended by the axioms $\generalizeprinciple$ and $\mixprinciple$---this is also reflected by  constant specification---and the iterated constant necessitation rule
\[
   \lrule{\jbox{c}_\agent \phi \in \CS}{\jbox{\star^n c}_{i_n}\ldots\jbox{\star c}_{i_1}\jbox{c}_\agent \phi}
 \]
for arbitrary agents $i_1,\ldots,i_n$. Here we employ the same term operator $\star$ as in the rule $\constnecrule^\star$ although the meaning of $\star$ in these two rules is a bit different.

\begin{theorem}[Internalization]
 Let\/ $\CS$ be an axiomatically appropriate constant specification. 
The system $\LPLTL^\mathsf{G}_\CS$ enjoys internalization.
 \end{theorem}
 \begin{proof}
 We proceed by induction on the derivation of $\phi$. There are two interesting cases:
  
 In case $\phi$ is $\lalways \psi$, derived using $\alwaysnecrule$, then, by the induction hypothesis, there is a term~$s$ such that $\jbox{s}_\agent \psi$ is provable.
 Now, we can use $\alwaysnecrule$ in order to obtain $\lalways \jbox{s}_\agent \psi$ and then $\generalizeprinciple$ and modus ponens to get $
  \jbox{\tgeneralize s}_\agent \lalways \psi
$.

In case $\phi$ is $\lnext \psi$, derived using $\nextnecrule$, then, as above, we obtain $\jbox{\tgeneralize s}_\agent \lalways \psi$.
Since $\CS$ is axiomatically appropriate, there is a constant~$c$ with $\jbox{c}_\agent (\Box \psi \to \lnext \psi)$.
Thus we finally conclude
$
  \jbox{c \tapp \tgeneralize s}_\agent \lnext \psi
$. 
 \end{proof}

 It is obvious how to formulate conditions on evidence functions that correspond to the additional principles of this section such that soundness results can be obtained, see~\cite{DBLP:journals/corr/Bucheli15}. However, it is not clear how to show the existence of such models and how to show completeness for these additional principles.  

 \section{Conclusions}\label{s:conclusion}

We introduced the temporal justification logic $\LPLTL_\CS$ and showed that it is sound and complete with respect to interpreted systems that are based on Fitting-models. To achieve this we had to adapt the usual canonical model construction of justification logic such that it yields a finite Fitting-model.
Further, we established that a suitable form of axiom necessitation can replace the necessitation rules for $\Box$ and $\lnext$ and thus make internalization possible.
Finally, we briefly discussed some additional principles that concern the interaction of knowledge, justifications, and time.

We finish this paper with some questions that show possible directions for future work.

\begin{question}
How does a temporal justification logic based on $\mathsf{JT45}$, i.e.~the justification counterpart of the modal logic~$\mathsf{S5}$, look like? 
The problem is that $\mathsf{JT45}$-models must satisfy the strong evidence condition, i.e.~for all $\system, r, n$ and each formula $\jbox{t}_\agent \phi$
\begin{equation}\label{eq:strong:1}
 \phi \in \evidence_\agent(r(n),t) \quad\text{implies}\quad
 (\system, r, n) \entails \jbox{t}_\agent \phi \,,
\end{equation}
see~\cite{Art08RSL,Pac05PLS,Rub06JLC,Stu13JSL}.
In infinite canonical models, the strong evidence property is an easy consequence of the Truth Lemma. In our temporal setting, we have a finite canonical model and the Truth Lemma is restricted to $\Sub^+(\chi)$. Hence it does not entail~\eqref{eq:strong:1} for all formulas~$\jbox{t}_\agent \phi$.
\end{question} 

\begin{question}\label{q:yav}
How can the typical examples, e.g., protocols related to message transmission, be formalized in $\LPLTL_\CS$?

Yavorskaya~\cite{TYav08TOCS} introduces multi-agent justification logics with interaction operations on the justification terms, in particular, she studies two principles:
 \begin{gather*}
\jbox{t}_i \phi \to \jbox{\tinspect_i^j t}_j\jbox{t}_i \phi 
\tag{\textsf{evidence verification}} \\
\jbox{t}_i \phi \to \jbox{\uparrow_i^j t}_j\phi 
\tag{\textsf{evidence conversion}}
 \end{gather*}
 where one agent's evidence is converted into another agent's evidence.
We believe that principles of this kind will be important in the context of this question.
For example, one might consider a temporal justification logic with principles such as
 \begin{gather*}
   \jbox{t}_\agent \phi \limplies \lnext \jbox{\textsf{sent}^\agent_j(t)}_j \phi \quad \text{or}\\
  \jbox{t}_\agent \phi \limplies \leventually \jbox{\textsf{sent}^\agent_j(t)}_j \phi \, .
 \end{gather*}
Here agent $i$ sends evidence $t$ for $\phi$ to agent $j$ and the term $\textsf{sent}^\agent_j(t)$ denotes the evidence that agent $j$ received for believing $\phi$.  
\end{question}

\begin{question}
 What happens if we require operations on justification terms to take time?

We could formalize this idea, e.g.,  by replacing  $\appax$, $\sumax$, and  $\posintax$ with
\begin{gather*}
\jbox{t}_\agent (\phi \limplies \psi) \limplies (\jbox{s}_\agent \phi \limplies \lnext\jbox{t \tapp s}_\agent \psi) \\
\jbox{t}_\agent \phi \vee \jbox{s}_\agent \phi \limplies  \lnext \jbox{t \tsum s}_\agent \phi \\
\jbox{t}_\agent \phi \limplies \lnext \jbox{\tinspect t}_\agent \jbox{t}_\agent \phi\,.
 \end{gather*}
 This might also relate to the logical omniscience problem~\cite{ArtKuz14APAL}.
 \end{question}

 \begin{question}
Can dynamic epistemic justification logics be translated into temporal justification logic  akin to~\cite{vDvdHR13}? 

There are several dynamic justification logics available, e.g., 
\cite{BucKuzStu11WoLLIC,BucKuzStu14,KuzStu13LFCS,Ren12Synthese}, which feature not only traditional public announcements but also specific forms of evidence based updates and evidence elimination.
It would be interesting to see what the relationship between those dynamic logics and temporal justification logic is.
 \end{question}
 
This paper showed a first successful combination of temporal and justification logic. While this initial work shows the feasibility of combining these logics with minimal interaction, the list of questions above shows that various interesting properties may arise from more intricate interactions between justified knowledge and time.



\end{document}